\documentclass[a4paper,10pt]{amsart}
\usepackage{graphicx}
\usepackage{amssymb,amsmath,amstext,amscd}

\theoremstyle{plain}
\newtheorem{thm}{Theorem}[section]
\newtheorem{proposition}{Proposition}[section]

\newtheorem*{main theorem}{Theorem}

\theoremstyle{definition}

\begin{document}

\title{An imbedding of spacetimes}
\author{Do-Hyung Kim}
\address{Department of Mathematics, College of Natural Science, Dankook University,
San 29, Anseo-dong, Dongnam-gu, Cheonan-si, Chungnam, 330-714,
Republic of Korea} \email{mathph@dankook.ac.kr}


\begin{abstract}
It is shown that any two-dimensional spacetimes with compact
Cauchy surfaces can be causally isomorphically imbedded into the
two-dimensional Einstein's static universe. Also, it is shown that
any two-dimensional globally hyperbolic spacetimes are conformally
equivalent to a subset of the two-dimensional Einstein's static
universe.

\end{abstract}

\maketitle

\section{Introduction} \label{section:1}

The group of symmetries plays central roles in theoretical physics
as well as many branches of mathematics. In this sense, to find or
to analyze the structure of symmetry group is important. In
\cite{CQG2}, it is shown that any two-dimensional spacetimes with
non-compact Cauchy surfaces can be causally isomorphically
imbedded into $\mathbb{R}^2_1$, and by use of this, in \cite{JGP},
it is shown that the groups of causal automorphisms on
two-dimensional spacetimes with non-compact Cauchy surfaces are
subgroups of that of $\mathbb{R}^2_1$. In other words, to imbed a
spacetime into a larger space with certain structure preserved is
important in analysis of structure groups.

In this paper, we show that any two-dimensional spacetimes with
compact Cauchy surfaces can be causally isomorphically imbedded
into two-dimensional Einstein's static universe. In conclusion,
any two-dimensional globally hyperbolic spacetimes can be causally
isomorphically imbedded into two-dimensional Einstein's static
universe. We also show that the role of causally isomorphic
imbedding can be replaced by conformal diffeomorphism.

\section{Preliminaries} \label{section:2}

In this section, we briefly review and improve the results in
\cite{CQG2} and \cite{CQG1}.

Let $M$ be a spacetime with a non-compact Cauchy surface $\Sigma$.
For $p \in J^+(\Sigma)$ and $q \in J^-(\Sigma)$, let $S_p^+ =
J^-(p) \cap \Sigma$ and $S_q^- = J^+(q)\cap \Sigma$. Then, the
compact and connected subsets $S_p^+$ and $S_q^-$ uniquely
determine the points $p$ and $q$, and thus $\textsl{C}_M^+ = \{
S_p^+ \,\, | \,\, p \in J^+(\Sigma) \}$ and $\textsl{C}_M^- = \{
S_q^- \,\, | \,\, q \in J^-(\Sigma) \}$ represent $M$.

If $S_p^+, S_q^+ \in \textsl{C}_M^+$ satisfy $S_p^+ \subset
S_q^+$, then we can show that $p \leq q$ and similar result also
holds for subsets in $\textsl{C}_M^-$. Also, it is obvious that
$S_p^+ \cap S_q^- \neq \emptyset$ if and only if $q \leq p$. In
this way, we can encode causal structures of $M$ into its Cauchy
surface $\Sigma$.

Let $N$ be a spacetime with a non-compact Cauchy surface
$\Sigma^\prime$ and $f : \Sigma \rightarrow \Sigma^\prime$ be a
homeomorphism. If $f$ induces a bijection from $\textsl{C}_M^+$
onto $\textsl{C}_N^+$, and from $\textsl{C}_M^-$ onto
$\textsl{C}_N^-$, then $f$ can be extended to a unique causal
isomorphism $\overline{f} : M \rightarrow N$. This is Theorem 5.4
in \cite{CQG1}.

If a homeomorphism $f : \Sigma \rightarrow \Sigma^\prime$ induces
a map, not necessarily bijective, from $\textsl{C}_M^+$ into
$\textsl{C}_N^+$, and from $\textsl{C}_M^-$ into $\textsl{C}_N^-$,
then $f$ can be uniquely extended to a causally isomorphic
imbedding $\overline{f} : M \hookrightarrow N$. In other words,
$\overline{f}$ is an imbedding and satisfies $x \leq y$ if and
only if $\overline{f}(x) \leq \overline{f}(y)$ for all $x$ and $y$
in $M$.

Let $\mathbb{R}^2_1 = \{ (t,x) \,\, | \,\, t,x \in \mathbb{R} \}$
be a two-dimensional Minkowski spacetime. Then $\mathbb{R}_{t_0} =
\{ (t_0, x) \,\, | \,\, x \in \mathbb{R} \}$ is a Cauchy surface
of $\mathbb{R}^2_1$. One of the characteristic properties of
$\mathbb{R}^2_1$ is that for any compact and connected subset $A$
of $\mathbb{R}_{t_0}$, there exist unique $p$ and $q$ such that
$S_p^+ = S_q^- = A$. Therefore, for any two-dimensional spacetime
$M$ with a non-compact Cauchy surface $\Sigma$, any homeomorphism
$f : \Sigma \rightarrow \mathbb{R}_{t_0}$ can be uniquely extended
to a causally isomorphic imbedding $\overline{f} : M
\hookrightarrow \mathbb{R}^2_1$ and thus we have the following.

\begin{thm} \label{imbedding}
For any given homeomorphism $f : \Sigma \rightarrow
\mathbb{R}_{t_0}$, we can extend $f$ to a causally isomorphic
imbeeding $\overline{f} : M \hookrightarrow \mathbb{R}^2_1$. In
other words, any two-dimensional spacetime with non-compact Cauchy
surfaces can be causally isomorphically imbedded in
$\mathbb{R}^2_1$.
\end{thm}
\begin{proof}
This is Theorem 5.1 in \cite{CQG2}.
\end{proof}

It is well-known that, if $f : M \rightarrow N$ is a causal
isomorphism with the dimension of $M$ bigger than two, then $f$
becomes a conformal diffeomorphism. (\cite{HKM}, \cite{Malament},
\cite{Fullwood}.) However, this is not the case when the dimension
is two and so the causal isomorphism given in the above theorem is
not necessarily a conformal diffeomorphism and not even
necessarily smooth. Nevertheless, even when the dimension is two,
if we take a $C^\infty$-diffeomorphism $f : \Sigma \rightarrow
\mathbb{R}_{t_0}$, then the extended map $\overline{f}$ is a
conformal diffeomorphism.

\begin{proposition} \label{confo-prop}
If we take $f : \Sigma \rightarrow \mathbb{R}_{t_0}$ to be
$C^\infty$ diffeomorphism, then the induced imbedding
$\overline{f}$ is a $C^\infty$-conformal diffeomorphism onto its
image in $\mathbb{R}^2_1$.
\end{proposition}
\begin{proof}
For $p \in J^+(\Sigma)$, there exist unique $x$ and $y$ in
$\Sigma$ such that $\partial S_p^+ = \{ x, y\}$ and, since $x$ and
$y$ are determined by unique null geodesics from $p$, the
dependence of $x$ and $y$ on $p$ is smooth. Since $f$ is
$C^\infty$, the dependence of $f(x)$ and $f(y)$ on $p$ is also
smooth, and by the same argument using null geodesics, the
dependence of $\overline{f}(p)$ on $f(x)$ and $f(y)$ is smooth.
Therefore, $\overline{f}$ is smooth and likewise,
$\overline{f}^{-1}$ is also smooth.

To show that $\overline{f}$ is conformal, it suffices to show that
$\overline{f}_*(v)$ is null whenever $v \in T_pM$ is a null
vector. Let $\gamma$ be a null geodesic such that $\gamma(0)=p$
and $\gamma^\prime(0)=v$. Then, since $\Sigma$ is non-compact,
there exist no null cut points along $\gamma$ and so we have
$\gamma(0) \leq \gamma(t)$ but not $\gamma(0) << \gamma(t)$ for
all $t$. Since $\overline{f}$ is a causal isomorphism, we have
$\overline{f}(\gamma(0)) \leq \overline{f}(\gamma(t))$ but not
$\overline{f}(\gamma(0)) << \overline{f}(\gamma(t))$. Therefore,
any causal curve from $\overline{f}(\gamma(0))$ to
$\overline{f}(\gamma(t))$ is a null pregeodesic and thus
$\overline{f}_*(v)$ is null. The same argument applied to
$\overline{f}^{-1}$ shows that $\overline{f}^{-1}$ is also
conformal.
\end{proof}

Therefore, we have the following.

\begin{thm} \label{conformal}
Let $M$ be a two-dimensional spacetime with non-compact Cauchy
surfaces. Then $M$ is conformally equivalent to a globally
hyperbolic open subset of $\mathbb{R}^2_1$ that contains $x$-axis
as a Cauchy surface.
\end{thm}
\begin{proof}
Take $t_0 = 0$ in the previous proposition.
\end{proof}

In fact, if sufficient smoothness is assumed, the study of causal
structure and the study of conformal structure are essentially the
same as the following theorem shows.

\begin{thm}
Let $f : M \rightarrow N$ be a conformal diffeomorphism. Then $f$
is either a causal isomorphism or an anti-causal isomorphism.
\end{thm}

\begin{proof}
Let $M$ and $N$ be spacetimes in which time-orientations are given
by nowhere-vanishing timelike vector fields $X$ and $Y$,
respectively. Let $Z = (f^{-1})_*(Y)$. Then, since $g_M(X,Z) =
g_N(f_*(X),Y)$, we can see that, if $g_M(X,Z) = 0$, then $F_*(X)$
is a spacelike vector, which contradicts to that $f$ is conformal.
Therefore, $g_N(f_*(X), Y)$ is nowhere zero and so, since $f_*(X)$
is timelike, $f_*(X)$ is either everywhere future-directed or
everywhere past-directed.

Let $v \in T_pM$ be a future-directed timelike vector. Then, since
$g_N(f_*(v), f_*(X)) = g_M(v, X_p) < 0$, $f_*(v)$ determines the
same time-orientation as that determined by $f_*(X)$. Therefore,
$f_*(v)$ is future-directed if $f_*(X)$ and $Y$ determine the same
time-orientation, and otherwise, $f_*(v)$ is past-directed.
\end{proof}

\section{Spacetimes with compact Cauchy surfaces} \label{section:3}

Let $(M,g)$ be a two-dimensional spacetime with compact Cauchy
surfaces. Then by Theorem 1 in \cite{Bernal}, $M$ is diffeomorphic
to $\mathbb{R} \times S^1$ and we can use $(t, e^{ix})$as a
coordinate on $M$.

Since $\exp : \mathbb{R} \rightarrow S^1$ defined by $\exp(x) =
e^{ix}$ is a covering map, the map $\pi : \mathbb{R} \times
\mathbb{R} \rightarrow M = \mathbb{R} \times S^1$ given by
$\pi(t,x) = ( t, e^{ix})$ is a covering map. Let $\overline{M}$ be
$\mathbb{R} \times \mathbb{R}$ with the pull-back metric $\pi^*g$.
Then, $\overline{M}$ is a universal covering space of $M$ in such
a way that $\pi$ is a time-orientation preserving covering map.
Then by Theorem 2.1 in \cite{JMP} or the proof of Theorem 14 in
\cite{Chernov}, $\overline{M}$ is globally hyperbolic with the
non-compact Cauchy surface $\Sigma = \{ (0,x) \,\, | \,\, x \in
\mathbb{R} \}$.

If we choose homeomorphism $f : \Sigma \rightarrow \mathbb{R}_0$
to be $f(x) = x$, where $\mathbb{R}_0 = \{ (0,x) \,\, | \,\, x \in
\mathbb{R} \}$ is a Cauchy surface of $\mathbb{R}^2_1$, then $f$
induces a causally isomorphic imbedding $\overline{f} : M
\hookrightarrow \mathbb{R}^2_1$.

The two-dimensional Einstein's static universe is $E = \mathbb{R}
\times S^1$ with the flat metric $-dt^2 + d\theta^2$. Thus the
universal covering space of $E$ is $\mathbb{R}^2_1$ with the
covering map $\pi_E(t, \theta) = (t, e^{i\theta})$. Therefore, we
have the following diagram.

\
\
\

If we identify $\overline{M} = \mathbb{R} \times \mathbb{R}$ with
$\overline{f}(\overline{M})$, then, since $f(x) = x$, we have $\pi
= \pi_E$ on $\overline{f}(\overline{M})$, and thus $M$ is causally
isomorphic to $\pi_E \circ \overline{f}(\overline{M})$, which is
an open subset of $E$. Therefore, we have the following.

\begin{thm} \label{cpt-imbedding}
Any two-dimensional spacetimes with compact Cauchy surfaces can be
causally isomorphically imbedded into the two-dimensional
Einstein's static universe.
\end{thm}

Since the homeomorphism $ f : \Sigma \rightarrow \mathbb{R}_0$
used in the above theorem is, in fact a diffeomorphism, by
Proposition \ref{confo-prop}, $\overline{f}$ is a smooth conformal
diffeomorphism and thus, since the covering maps $\pi$ and $\pi_E$
are smooth local isometries, we have the following.

\begin{thm} \label{cpt-conf}
Any two-dimensional spacetimes with compact Cauchy surfaces can be
conformally equivalently imbedded into the two-dimensional
Einstein's static universe.
\end{thm}

It is a well-known fact that two-dimensional Minkowski spacetime
can be causally isomorphically, or conformally equivalently
imbedded into the two-dimensional Einstein's static universe in
such a way that $\mathbb{R}_0$ is sent to $S^1-\{\mbox{a point}\}$
(See sectin 5.1 in \cite{HE}). Then, by combining this with
Theorem \ref{imbedding} or Theorem \ref{conformal}, we have the
following.

\begin{thm}
Any two-dimensional spacetimes with non-compact Cauchy surfaces
can be causally isomorphically, or conformally equivalently
imbedded into the two-dimensional Eisntein's static universe.
\end{thm}

If we now combine this theorem with Theorem \ref{cpt-imbedding}
and Theorem \ref{cpt-conf}, we have the following.

\begin{thm}
Any two-dimensional globally hyperbolic spacetimes can be causally
isomorphically, or conformally equivalently imbedded into the
two-dimensional Einstein's static universe.
\end{thm}

In this theorem, it must be noted that if spacetimes have
non-compact Cauchy surfaces,then the Cauchy surface is sent to
$S^1-\{\mbox{a point}\}$ in $E$ and, if they have compact Cauchy
surfaces, then the Cauchy surface is sent to $S^1$.

\section{Acknowledgement}

The present research was conducted by the research fund of Dankook
university in 2013.


\begin{thebibliography}{999}
%
\bibitem{CQG2} D.-H. Kim,
{\it An imbedding of Lorentzian manifolds}, Class. Quantum. Grav.
{\bf 26}, (2009) 075004.

%
\bibitem{JGP} D.-H. Kim,
{\it A classification of two-dimensional spacetimes with
non-compact Cauchy surfaces}, J. Geom. Phys. {\bf 73}, (2013)
252-259.

%
\bibitem{CQG1} D.-H. Kim,
{\it A note on non-compact Cauchy surfaces}, Class. Quantum. Grav.
{\bf 25}, (2008) pp. 238002.



%
\bibitem{HKM} S.W. Hawking, A.R.King and P.J. McCarthy,
{\it A new topology for curved space-time which incorporates the
causal, differential, and conformal structures}, J. Math. Phys.
{\bf 17}, (1976) 174.

%
\bibitem{Malament} D.B. Malament,
{\it The class of continuous timelike curves determines the
topology of spacetime}, J. Math. Phys. {\bf 18}, (1977) 1399.

%
\bibitem{Fullwood} D.T. Fullwood,
{\it A new topology on space-time}, J. Math. Phys. {\bf 33},
(1992) 2232.


%
\bibitem{Bernal} A. N. Bernal and M. S$\acute{a}$nchez,
{\it On smooth Cauchy hypersurfaces and Geroch's splitting
theorem}, Commun. Math. Phys. {\bf 243}, (2003) 461.

%
\bibitem{JMP} D.-H. Kim,
{\it Lorentzian covering space and homotopy classes}, J. Math.
Phys. {\bf 52}, (2011) 122501.

%
\bibitem{Chernov}, V. V. Chernov and Y. B. Rudyak,
{\it Linking and causality in globally hyperbolic space-times},
Commun. Math. Phys., {\bf 279}, (2008) 309.


%
\bibitem{HE}, S. Hawking and J. Ellis,
{\it The large scale structure of spacetime}, Cambridge Univ.
Press, (1973) 309.


\end{thebibliography}
\end{document}